\documentclass{llncs}
\usepackage{makeidx,newicktree}
\usepackage{rotating}
  \usepackage{graphicx}
  \usepackage{psfrag}
  \usepackage{amsfonts}
\usepackage{amsmath,amssymb,color}
\usepackage{url}

  \usepackage{amsfonts}
\graphicspath{{figs/}}

\newcommand{\R}{\mathbb R}

\providecommand{\abs}[1]{\lvert#1\rvert}

\def\comment#1{\textit{[#1]}}
\def\comment#1{}

\begin{document}
\title{First steps toward the geometry of cophylogeny}

\author{Peter Huggins\inst{1}, Megan Owen \inst{2}, and Ruriko Yoshida\inst{3}}

\institute{Lane Center for Computational Biology, Carnegie Mellon University, Pittsburgh, PA, USA \and Statistical and Applied Mathematical Sciences Institute, Research Triangle Park, NC, USA \and University of Kentucky, Lexington, KY, USA}

\maketitle

\begin{abstract}
Here we introduce researchers in algebraic biology to the exciting new field of {\em cophylogenetics}.  
Cophylogenetics is the study of concomitantly evolving organisms (or genes), such as host and parasite species.  Thus the natural objects
of study in cophylogenetics are tuples of related trees, instead of individual trees.  
We review various research topics in algebraic statistics for phylogenetics, and propose analogs for cophylogenetics.
In particular we propose {\em spaces of cophylogenetic trees}, {\em cophylogenetic reconstruction}, and {\em cophylogenetic invariants}.
We conclude with open problems. 
\end{abstract}

\section{Introduction}
%At the heart of phylogenetics is the prevailing biological concept is that the diversification of
%life forms is caused by random mutations and the separation of gene pools.  

%Traditionally, mathematical phylogenetics has studies the evolutionary relationships between species.   
%Phylogenetics studies the evolution species, specifically clades of species which have descended from a common ancestor.

%{\bf CITE LIOR AND RADU:} \cite{Mihaescu2006}

Phylogenetics has provided an abundant source of applications for algebraic statistics, with research areas including phylogenetic invariants, the geometry of tree space, 
and analysis of phylogenetic reconstruction.  Traditionally these applications --- like phylogenetics at large --- focused on the common ancestries among a single set of species 
(or set of gene homologs).  

On the biological front, however, phylogenetic research has since expanded to include other types of evolutionary relationships besides common ancestry.  
Here we present one of the largest new research topics in phylogenetics, called {\em cophylogenetics}, and we explore possible applications of algebraic statistics.  
Just as phylogenetics can be loosely described as the study of evolution, cophylogenetics is essentially the study of {\em coevolution}.   
Coevolution is the concomitant evolution and speciation of one species (or gene) with another.  In biology there are two actively studied examples of coevolution:

\begin{itemize}
\item Host--parasite coevolution (or more generally, symbiont coevolution):  Symbiont species interact with one another and often migrate together, and thus tend to have parallel lineages 
during evolution and speciation.   Thus symbionts often have similar phylogenetic trees.
\item  Gene trees and species trees:  Loosely speaking, genes can be thought of as ``symbionts'' living within a species.  Thus gene trees are often similar to one another and the species tree.
\end{itemize}

%Deviations from parallel evolution that are increasingly 
%recognized in pathogen and human genomes include 
%gene duplications, lateral gene transfers between species, retention of ancestral polymorphisms by 
%balancing selection, and accelerated evolution by neofunctionalization. 

There have been many studies of coevolution of clades of host species and their corresponding clades of parasite species
(\cite{Pages2003} and its references).
One well-known example is the set of gopher/louse pairs reported in \cite{Hafner90}.
\begin{figure}[ht!]
\center
\includegraphics[scale= 0.6]{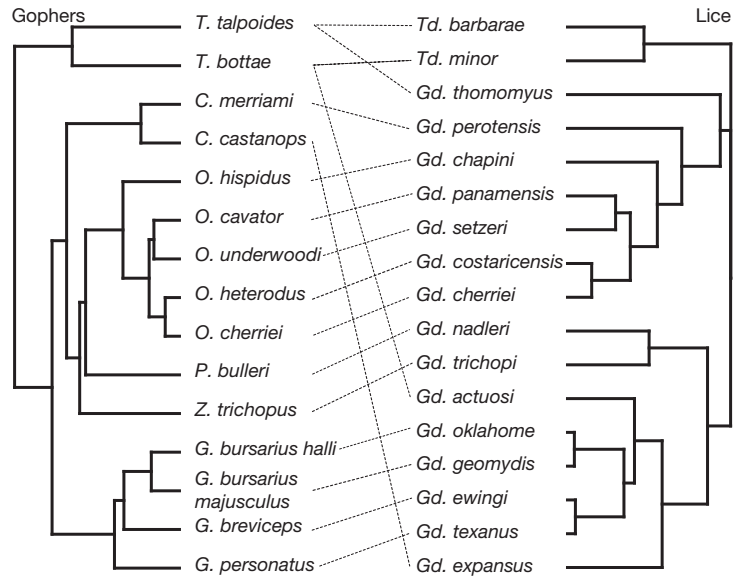}
\caption{Phylogenetic trees for gopher and louse data sets \cite{Hafner90} constructed via {\tt BEAST}. Hosts and their parasites are indicated by connecting dashed lines.
% Genera: {\em O. = Orthogeomys, Z. = Zygogeomys, P. = Pappogeomys, C.
% = Cratogeomys, G. = Geomys, T. = Thomomys, Gd. = Geomydoecus, Td = Thomomydoecus.}
}\label{fig2}
\end{figure}
Even though there is significant evidence of parallel evolution between gophers and
their lice, their reconstructed tree topologies differ (Fig. \ref{fig2}).
In fact, reconstructed host and parasite trees are rarely identical.  This disagreement  could be due to reconstruction errors, either caused by noise in the input data or heuristical reconstruction methods, or host and parasite trees could be truly different.

As host and parasite coevolve, there are six commonly recognized types of events which can occur along lineages \cite{Pages2003}.  
These are shown in Fig. \ref{fig1}.
\begin{enumerate}
\item[(a)] A host and a parasite cospeciate, i.e., they speciate together.
\item[(b)] A parasite changes its host (\emph{host switch}), which is equivalent 
to a gene transfer in gene trees. 
\item[(c)] A parasite speciates independently of their host.
\item[(d)] A parasite goes extinct.
\item[(e)] A parasite fails to colonize all descendants of a speciating host lineage.
\item[(f)] A parasite fails to speciate.
\end{enumerate}
%Biologically, these processes in a host--parasite association occur because: 
%host and parasite cospeciate (a), or the parasite might speciate independently (b, c);  one or more of the descendant parasites may colonize a new host (b), 
%or the parasite may remain on the original host (c);  
%absence of a parasite from a host where it would be expected to occur may be due to extinction of that parasite 
%(d), or the ancestors of the host lineage may have not inherited the ancestral parasite (e); 
%and a host may speciate independently of their parasites so that two hosts share the same parasite (f) \cite{Pages2003}.  
Notice that in the figure, events (b) through (f) can cause host and parasite tree topologies to differ.  
%except the first one may result in different tree topologies for the host and parasite trees.
\begin{figure}[!ht]
\begin{center}
\includegraphics[scale=0.6]{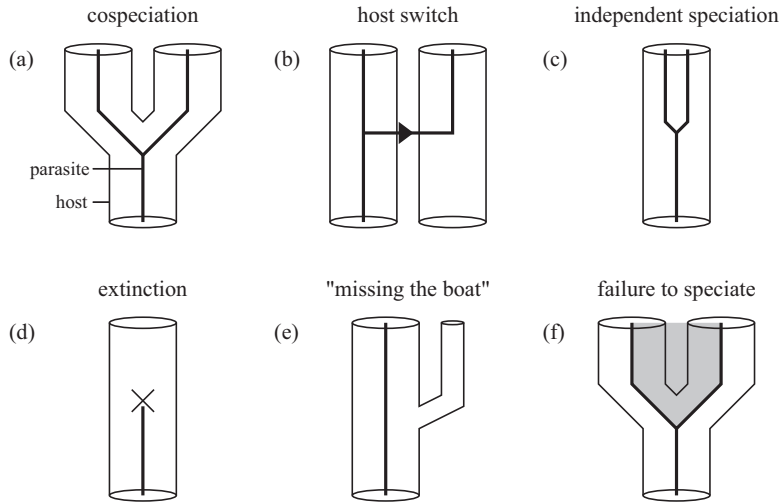}
\end{center}
\caption{Evolutionary events which can occur during host--parasite coevolution. \cite{Pages2003}.}\label{fig1}
\end{figure}

%Host-switching, depicted in (b) in the Figure, is perhaps the most interesting process during host/parasite coevolution.  
%When hosts and parasites cospeciate, some parasites might remain compatible with nearby hosts in the host tree.  
%Occasionally a parasite species $p$ associated with host $h$ might encounter a related host $h'$ after a period of separation.  
%In some cases $p$ can invade $h'$ and replace its parasite $p'$.  Such an event is called a {\em host switch}.  

Analagous to host/parasite relationships, gene trees for a given set of species can vary from gene to gene and also differ from the species tree.  
In statistics, the relation between gene and species trees is well-understood in terms of coalescent processes \cite{Hein2005}.
However coalescent models usually assume that genes cannot be transfered between members of different species.  In reality, microbial organisms, for example, can exchange genetic material in a process called {\em lateral gene transfer} \cite{Doolittle1999}, which is analogous to host switching.  Just as host switching can cause parasite trees to disagree with host trees, lateral gene transfer can cause gene trees to disagree 
with species trees.  Combinatorially, these mechanisms correspond to {\em subtree prune and regraft (SPR)} operations \cite{Semple2003}, which are discussed further in Section \ref{cophy}.

%There are many heuristical and statistical techniques that have been developed to detect similarities and differences both between host and parasite trees and gene trees.
%(See section \ref{supp} for references).

Many techniques have been developed to compare gene trees \cite{Liu2007,Edwards2007,Ane2007,tmp,PHT,HK,geneLRT,PLS}, and host and parasite trees \cite{Brook2003,chris2008,Huelsenbeck2000b,Hafner90}.

%{\bf TODO:  LOOK AT BROOK2003} \cite{Brook2003} {\bf AND SEE WHAT THE HELL THEY DO}

%Bayesian estimation methods \cite{Liu2007,Edwards2007,Ane2007}, the Templeton test implemented in {\tt paup*} 
%\cite{Swofford1998} (e.g., 
%\cite{tmp}), the partition-homogeneity test (PHT) also implemented with 
%{\tt paup*} \cite{PHT},  Kishino-Hasegawa test 
%\cite{HK}, and the likelihood ratio test (LRT)  \cite{geneLRT}
%These methods are also called partition likelihood support (PLS) 
%\cite{PLS}.  On the
%other hand, most methods used for the host-parasite analysis 
%test whether there is a significant level of {\em congruence} between the trees.
%Since \cite{Henning1966}, 
%there have been many studies analyzing host-parasite cospeciation (\cite{Brook2003} and its references, \cite{chris2008}).
In existing methods for comparing host and parasite trees, first point estimates of trees are made and then trees are compared to one another.  
However estimating each tree separately can exagerate the 
true differences between the trees.  Thus one of our main tenets in this paper is that when reconstructing and studying related trees, researchers in mathematical phylogenetics 
should explicitly consider tuples of trees, which we call cophylogenies:  

\begin{definition}\label{def0}
Let $\mathcal{T}_{H}$ and $\mathcal{T}_{P}$ be spaces of trees for two sets of taxa $H$ and $P$.  A {\em cophylogeny} is any pair of
trees $(T_H, \,  T_P) \in \mathcal{T}_H \times \mathcal{T}_P$.  
\end{definition}

More generally, for tree spaces $\mathcal{T}_{G_1}, \ldots, \mathcal{T}_{G_\ell}$,
a cophylogeny is any tuple $({T}_{G_1}, \ldots, {T}_{G_\ell}) \in \mathcal{T}_{G_1} \times \cdots \times \mathcal{T}_{G_\ell}$.  However in this paper 
we will focus mainly on the case of two sets of taxa, $H$ and $P$.  
%We will assume throughout that trees $T_P, T_H$ are defined on the same number of taxa $n$.  

%For example, the support of any joint posterior distribution 
%$P(T_H, T_P |H,P)$ on
%$\mathcal{T}_H \times \mathcal{T}_P$ which satisfies
%\[
% P(T_H, T_P | H,P) \neq P(T_H | H,P) \cdot P(T_P | H,P)
%\]
%is a set of cophylogenies.   
%%One extreme case of cophylogeny is perfect
%codivergence, that is, assuming only a process (a) occurs.  Then the set of cophylogenies is the subset in $\mathcal{T}_H \times \mathcal{T}_P$ such that $T_H$ and $T_P$ are equal (for e.g., 
%\cite{Huelsenbeck2000b}).  

There has been much study on underlying combinatoric, algebraic, and 
polyhedral geometric structures
for phylogenetic trees defined on a fixed set of species (see \cite{ASCB} and its references).
However for cophylogenies --- particularly cophylogeneis of trees which are presumed to be related --- the combinatoric, algebraic, and polyhedral geometric 
structures have received little attention.  Thus we propose extending mathematical phylogenetics to include  
cophylogenetics.  To this end we introduce {\em spaces of cophylogenetic trees}, {\em cophylogenetic reconstruction} problems, {\em cophylogenetic invariants}, and new geometries of tree space.  

%{\bf FIX ME LAST!!!} This paper is organized as follows:  {\bf ????} We conclude in Section \ref{openprobs} with open problems.

\section{Spaces of cophylogenetic trees}\label{cophy}
%Through out this paper we assume that there is no host switching or gene transfer (i.e., the case (b) in Figure \ref{fig1}) and we consider the host tree
%vs their parasite tree or the species tree vs gene trees.   
%Also we assume that in the evolution history 
%only one divergence of a gene or a speciation occurs at a time.

In this paper, we assume all trees are unrooted unless specified.  
Nevertheless, almost all results will also be directly applicable to rooted trees 
e.g. by attaching a designated ``root leaf'' to convert rooted trees into unrooted trees.  
We also assume that all trees have $n$ leaves, which will usually be labeled $1,2,\ldots,n$.

\begin{definition} 
A dissimilarity map on $\{ 1, 2, \ldots ,n \}$ is an $n \times n$ symmetric matrix $D=\{d_{ij}\}_{i,j=1}^n$
 with zeroes on the diagonal and all other entries positive.  
\end{definition}
Equivalently the set of dissimilarity maps is $\R_{+}^{{n \choose 2}}$.  
%We use the notation $D=\{d_{ij}\}_{i,j=1}^n$ because then $d_{ij}$ can be regarded as the ``distance'' 
%between $i$ and $j$, although $D$ does not necessarily define a true metric on $\{1,2,\ldots,n\}$ 
%because the distances $d_{ij}$ might not satisfy the triangle inequality.  

\begin{definition} \label{def0b}
Let $D$ be a dissimilarity map.
$D$ is a {\em tree metric} if 
there exists an (edge-weighted) tree $T$ with leaves $\{1, 2, \cdots, n\}$, such that
\begin{itemize}
\item All edge weights in $T$ are positive, 
\item For every pair of $i, \, j$, $d_{ij} = $ the sum of the edge weights along
 the path from $i$ to $j$.
\end{itemize}
%\noindent
%Also we call such $T$ an {\em additive tree}.
\end{definition}
Equivalently tree metrics can be defined by the {\em Four Point Condition}.
\begin{theorem}[Four Point Condition \cite{Buneman1971}]\label{thm1}
Let $D$ be
a dissimilarity map.  Then $D$ is a tree metric if and only if for all possible distinct leaves $i,j,k,l$, the maximum of 
$\{ d_{ij} + d_{kl}, d_{ik}+d_{jl},  d_{il}+d_{jk} \}$ is achieved at least twice.
\end{theorem}

By Theorem \ref{thm1}, the set of tree metrics on $\{1,2,\ldots,n\}$ can be realized as a union of cones in $\R_{+}^{n \choose 2}$ \cite{ASCB}.
%Note the mapping $T \to D = D(T)$ in Definition \ref{def0b} is injective from unrooted trees with positive edge weights to tree metrics. 
%In this paper, the space of host trees $\mathcal{T}_{H}$ and
%the space of parasite trees $\mathcal{T}_{P}$ mean the unions of cones
%in the spaces of dissimilarity maps defined by equations and inequalities 
%from the Four Point Condition. 

%{\bf Actually, there is not a 1-1 map between dissimilarity maps at multifurcating and/or weighted rooted trees.
%For example, for weight rooted trees, consider the tree with just two leaves coming out of the root.  In one tree,
% the edge to leaf A is weighted by 1 and the edge to leaf B is weighted by 2.  In the second tree, the edge weights
% are reversed.  Then the two trees have the same dissimilarity map, assuming we don't tree the root as a leaf (but if 
% we did, we can not get distance to it from the sequence data.). }

%Now we are ready to define {\em spaces of cophylogenetic trees}.  
\begin{definition}\label{def00}
A subset $S \subset \mathcal{T}_H \times \mathcal{T}_P$ is called a {\em space of cophylogenetic trees}.
If $\mathcal{T}_H$ and $\mathcal{T}_P$ are spaces of tree topologies (instead of tree metrics), then 
$S \subset \mathcal{T}_H \times \mathcal{T}_P$ is called a {\em space of cophylogenetic tree topologies}.
%Given a space of cophylogenetic trees, the projection $S_{T_H} = \{T_P \, : \,  (T_H, T_P) \in S \} \subset \mathcal{T}_P$  is called
%the {\em space of cophylogenetic trees given $T_H$}.
\end{definition}
%\begin{remark}\label{rk1}
%In general $S_{T_H} \not = \mathcal{T}_{P}$ and $S \not = \mathcal{T}_{H} 
%\times \mathcal{T}_{P}$.  
%\end{remark}
%We illustrate Remark \ref{rk1} with the following example.
Our definition is deliberately vague, as there are many spaces of cophylogenetic trees 
which are biologically and mathematically interesting.

%\begin{example}\label{ex1}
%If we assume perfect codivergence, that is, $T_H$ and $T_P$ are identical (for e.g.,
%\cite{Huelsenbeck2000b}), the space of cophylogenetic trees is the diagonal of $\mathcal{T}_H \times \mathcal{T}_P$ :  
%\begin{align*}
%S = \{(D_H, D_P):  D_H \text{ is a tree metric for }T_H \text{ and } D_P = D_H\}
%which is not equal to $\mathcal{T}_{H} \times \mathcal{T}_{P}$.
%\end{align*}
%\qed
%\end{example}

%\subsection{Examples} \label{example}
\begin{example}\label{ex1b}
Even if trees $T_H$ and $T_P$ have the same topology, variable rates of evolution 
can cause edge lengths to differ between the trees.  Thus we can consider the ``topology diagonal'' of $\mathcal{T}_H \times \mathcal{T}_P$: 
\begin{align*}
S &= \{(D_H, D_P):  D_H \text{ and } D_P \text{ are tree metrics, with the same underlying } \\
&\text{ bifurcating tree topology }\}.
\end{align*}
Analagous to the traditional space of trees, the topology diagonal is a union of $(2n-3)!$ polyhedral cones, and can be defined by an extended Four Point Condition:  
\begin{proposition}
If $D_H, D_P \in \R^{n \choose 2}$ are dissimilarity maps, then  $(D_H, D_P) \in S$ if and only if the Four Point Condition holds for the three dissimilarity maps $D_H, D_P$, and $D_H + D_P$
(where each maximum in the Four Point Condition is attained exactly twice).  Thus the topological closure of $S$ is (the negation of) a tropical variety.  
\end{proposition}
\begin{proof}
It suffices to prove that if $D_H, D_P$ are tree metrics with respective bifurcating tree topologies $T_H,T_P$, then $D_H + D_P$ is a tree metric if and only if $T_H = T_P$.  So suppose 
$D_H = \{d_{ij}\}, D_P = \{c_{ij}\}$
are tree metrics with underlying bifurcating tree topologies $T_H, T_P$.  If $T_H = T_P$, then $D_H + D_P$ is also a tree metric with tree topology $T_H,$ because tree metrics with topology $T_H$ form a cone.
Conversely, if $T_H \neq T_P$, then there will be some choice of four taxa $i,j,k,l$, 
such that the quartet induced by $i,j,k,l$ in $T_H$ is different than the quartet induced by $i,j,k,l$ induced in $T_P$.  By the Four Point Condition, in the matrix 
\begin{align*}
\begin{array}{ccc}
 d_{ij} + d_{kl} \quad & d_{ik}+d_{jl} \quad & d_{il}+d_{jk} \\
 c_{ij} + c_{kl} \quad & c_{ik}+c_{jl} \quad  & c_{il}+c_{jk} \\
\end{array}
\end{align*}
the maximum of each row will be attained twice---in fact, exactly twice since $T_H,T_P$ are bifurcating.  Thus each row attains a strict minimum.  
Furthermore, since $i,j,k,l$ induce different quartets in $T_H$ and $T_P$, the row minimums must be in different columns.  Thus, without loss of generality we can write the above matrix as 
\[
\begin{array}{ccc}
x \quad & x \quad & y \\
w \quad & z \quad & z \\
\end{array}
\]
where $y < x$ and $w < z$.  Summing the rows, we get the three numbers $(x+w,x+z,y+z)$, which attains a unique maximum $x+z$.  Thus $D_H + D_P$ cannot satisfy the Four Point Condition so $D_H + D_P$ 
is not a tree metric.
\qed
\end{proof}
Tropical varieties were first introduced to mathematical phylogenetics in \cite{speyer}, where it was shown that the space of trees is a tropical variety.  
We think it is an interesting problem to find other important spaces of cophylogenetic trees which are tropical varieties, or can be expressed by conditions involving linear equations and inequalities.
%{\bf THROW IN A COMMENT OR SOMETHING}
\end{example}

\begin{example}[Host switching/lateral gene transfer]
As we noted earlier, host-switching and lateral gene transfers correspond
to subtree-prune-and-regraft (SPR) \cite{Semple2003} operations on trees.
In an SPR operation, a subtree is detached, or pruned, from the tree by
cutting an edge, and reattached to the middle of a different edge.
The {\em SPR distance} between two trees is the minimum number of SPR
operations needed to transform one tree into the other.
As SPR operations are fundamental in mathematical cophylogenetics, we
define the {\em $k$-SPR space of cophylogenetic trees} as the set of all
cophylogenies $(T_H, T_P)$ where $T_H, T_P$ have SPR distance no more than
$k$.  If $k=1$, there are $2(n-3)(2n-7)$ tree topologies for $T_H$, 1 SPR
operation from $T_P$, as described in \cite[Theorem 2.6.2]{Semple2003}.

However, for unrooted trees, the SPR distance only provides a lower bound
on the number of host-switches or lateral gene transfers that have
occurred, because in this case, the SPR operations must be consistent with
the trees' orientations with respect to time.  To find the minimum number
of host-switches or lateral gene transfers, we must consider the SPR
distance between rooted trees, whose vertices are totally ordered with
respect to time \cite{Song2003}.  See \cite{Song2006} for a study of the analagous 
1-SPR space of cophylogenetic trees in this setting.  
\qed
\end{example}

\begin{example}[Coalescent cophylogeny]
In a coalescent cophylogeny, $T_H$ is a (rooted) species tree, and $T_P$ is
a (rooted) gene tree generated from $T_H$ according to the coalescent
process \cite{Hein2005}.  A \emph{coalescent history} is a list of the
branches of the species tree on which coalescences in the gene tree occur.
 We define the \emph{k-coalescent history space} to be the set of all
cophylogenies  $(T_H, T_P)$ for which the coalescent history for $T_H$ and
$T_P$ occurs with probability $\geq k$ out of all valid coalescent
histories for $T_H$.  Degnan and Salter \cite{Degnan2004} showed how to
compute this probability.  Note that the more dissimilar the topology of
$T_H$ and $T_P$, the smaller the number of valid coalescent histories, and
the lower their probabilities.  However, for any given species tree $T_H$,
there is at least one valid coalescent history for each gene tree
topology, namely, the history in which all coalescent events occur before
any speciation in $T_H$.
\qed
\end{example}

\begin{example}
A \emph{Nearest Neighbor Interchange (NNI)} operation swaps two adjacent subtrees, or subtrees joined by one edge, in an unrooted tree, as illustrated in Fig.~\ref{fig:NNIoperation}
%If the tree is rooted, this operation is called a \emph{rotation}. 
The \emph{NNI distance} between two trees is the minimum number NNI operations needed to transform one tree into the other \cite{Robinson1971}.

Analagous to the $k$-SPR space of cophylogenetic trees, we can define a space of cophylogenetic trees as the set of all cophylogenies  $(T_H, T_P)$ 
where $T_H,T_P$ have NNI distance no more than $k$.  
In general, we can define a space of cophylogenetic trees as the set of all cophylogenies $(T_H,T_P)$ where $d(T_H, T_P) \leq k$ for some distance or disimilarity measure on tree metrics (or tree topologies).

\begin{figure}[ht]
\centering
\includegraphics[scale = 0.35]{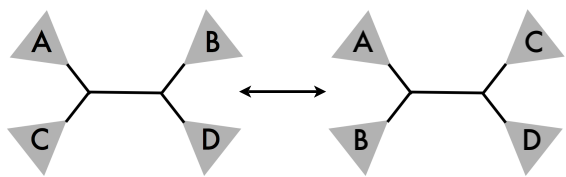}
\caption{An NNI operation.}
\label{fig:NNIoperation}
\end{figure}
Other choices for $d(T_H,T_P)$ might include the geodesic distance between $T_P$ and $T_H$ in tree space \cite{Billera2001}; the quartet distance \cite{Estabrook1985}; and the Robinson-Foulds symmetric difference \cite{RF81}.
%\begin{itemize}
%\item $d(T_H, T_P) := $ geodesic distance between $T_P$ and $T_H$ in tree space  \cite{Billera2001}
%\item $d(T_H, T_P) := $ quartet distance (i.e., the number of quartets disagreeing between two trees)
%\item $d(T_H, T_P) := $ subtree prune and regraft (SPR) distance 
%\item $d(T_H, T_P) := $ Robinson and Foulds symmetric difference \cite{RF81}
%\end{itemize}
\qed
\end{example}

\begin{example}[$k$-interval cospeciation]
In host--parasite coevolution, a speciation in the host is likely to be followed by a reactionary speciation in the parasite, and vice versa.  
If a reactionary speciation is delayed long enough, then host and parasite tree topologies can disagree.  In fact, if multiple consecutive speciations occur in host before a reactionary speciation 
occurs in parasite (or vice versa), then the tree topologies can be quite different.  

Biologically, it is highly unlikely that a large number of consecutive speciations can accumulate in a host lineage, without any reactionary speciation in parasite.  
Thus, when reconstructing host and parasite trees, we might assume that only a bounded number of consecutive speciations can occur in any host lineage before a reactionary speciation in parasite 
(and vice versa).  Combinatorially this implies that for each pair of host species $A,B$, and corresponding parasite species $a, b$, the number of edges between $A,B$ 
is within $k$ of the number of edges between $a,b$.  We say such a cophylogeny satisfies \emph{$k$-interval cospeciation}, and the set of all such cophylogenies is called the \emph{$k$-interval 
space of cophylogenetic trees}.  
\qed
\end{example}

If we treat $T_H$ and $T_P$ as having edge weights of 1, then asking whether $T_H$ and $T_P$ satisfy $k$-interval cospeciation is 
equivalent to asking whether the $L_{\infty}$-norm between their dissimilarity maps is $\leq k$.  
Although other norms on dissimilarity maps have been studied (\cite{Williams1971}, \cite{Steel1993}), the $L_{\infty}$-norm appears to have only been used in the context of a map between two versions of tree space \cite{Moulton1999}. 

We believe that in cophylogenetics applications, distances based on $k$-interval cospeciation will be more useful than the NNI distance.  
There has been some work characterizing trees within a prescribed NNI distance \cite{Robinson1971}, but to our knowledge the combinatorial properties of $k$-interval cospeciation
have not been studied.  We note that $k$-interval cospeciation and NNI distance are related concepts, and for $k = 1$ we have 

\begin{theorem}\label{thm3}
Suppose $T_H, T_P$ are unrooted trees on the same set of leaves.  Then $T_H$ and $T_P$ satisfy 1-interval cospeciation if and only if $T_H$ and $T_P$ differ by at most 1 NNI operation.
\end{theorem}
%That if $T_H$ and $T_P$ differ by at most 1 NNI, then they satisfy 1-interval cospeciation can be seen by observation.  
%To show the opposite direction, we fix one NNI, and show by case analysis that any additional sequence of NNI will result in the number of edges between some pair of leaves changing by at least two.  
%Due to space constraints, a proof of Theorem \ref{thm3} is posted in \url{http://polytopes.net/pdf/proof.pdf}. 

\emph{Proof}:  See Section \ref{sec:k_1_proof}.

Since there are $2n -6$ bifurcating trees which are exactly one NNI move away from a given bifurcating tree with $n \geq 4$ leaves \cite{Robinson1971}, 
the $1$-interval space of cophylogenetic trees on $n$ taxa contains 
$(2n - 5) \cdot (2n - 5)!!$ ordered pairs of bifurcating tree topologies when $n \geq 4$.

\section{Cophylogenetic reconstruction} \label{algorithms}

In the popular viewpoint echoed in \cite{Moulton1999}, distance-based methods for reconstructing phylogenetic trees from dissimilarity maps can be regarded as {\em retractions} from $\R_+^{n \choose 2}$ to 
tree metrics.  Due to their rich mathematical structure, two methods for phylogenetic reconstruction have received a great deal of attention in the mathematical biology community:   
Neighbor joining (NJ) \cite{Saitou1987,Mihaescu2006} and balanced minimum evolution (BME) \cite{Desper2002}.  
Intriguingly, Gascuel and Steel \cite{Steel2006} have recently shown that NJ is a greedy heuristic for building BME trees; see also \cite{NJME}.   

In this paper we propose distance-based {\em cophylogenetic reconstruction}, which means infering a cophylogeny $(T_H, T_P)$ given input tuples of dissimilarity maps $(D_H, D_P)$.  
Common methods for reconstructing $T_H$ and $T_P$ merely perform standard phylogenetic reconstruction, infering $T_H$ from $D_H$, and $T_P$ from $D_P$.  
Ideally, methods for cophylogenetic reconstruction should 
incorporate constraints or mixed objective functions that account for similarity between $T_H$ and $T_P$.  
We believe cophylogenetic reconstruction is an important avenue of future research in mathematical phylogenetics. 

Specifically we propose studying {\em constrained cophylogenetic reconstruction}, and {\em minimum coevolution} methods.  
Due to the widespread mathematical interest in NJ and BME, we will focus on cophylogenetic reconstruction based on these methods.  
See \cite{NJME} for a description and references on NJ and BME; here we will only briefly introduce basic facts and notation for the BME method.
Given a dissimilarity map $D$, where $D$ is assumed to be a noisy observation of a tree metric $D_T$, BME chooses a tree topology $T$ whose sum of estimated branch lengths is minimal.  
The sum of estimated branch lengths is also called the {\em length} of $T$, written $\ell(T)$.  The formulation of BME has a rich and elegant structure:  for each topology $T$ 
we have $\ell(T) = b(T) \cdot D$, where $b(T)$ is the {\em BME vector} for $T$ (which does not depend on $D$).  
Thus, BME is equivalent to minimizing the linear functional $D$ over the polytope $B = \hbox{conv} \{ b(T) \}_T$ , which is called the {\em BME polytope} \cite{NJME}.

\subsection{Retraction onto spaces of cophylogenetic trees}
Fix a space of cophylogenetic trees $S \subset {\mathcal T}_H \times {\mathcal T}_P \subset \R_+^{n \choose 2} \times \R_+^{n \choose 2}$.  
A {\em constrained cophylogenetic reconstruction} method is a retraction  
$\R_+^{n \choose 2} \times \R_+^{n \choose 2} \to S$, i.e. a mapping which restricts to the identity map on $S$.  We believe retractions onto spaces of cophylogenetic trees 
are an important new field of study in mathematical biology.  

In particular, we can formulate {\em constrained joint BME}:  Given dissimilarity maps $D_H,D_P$, find a pair of tree topologies $(T_H,T_P) \in S$ 
whose length sum $\ell(T_H) + \ell(T_P)$ is minimal (where tree length $\ell(T)$ is defined as in BME).  Simiarly we can define the {\em joint BME polytope}, which 
is the subpolytope $B' \subset B \times B \subset \R_+^{n \choose 2} \times \R_+^{n \choose 2}$, whose vertices correspond to pairs of tree topologies $(T_H,T_P) \in S$.  

Constrained joint BME is a very new research topic.  Recently Matsen \cite{erick} has studied constrained joint BME under the $k$-SPR space of rooted cophylogenetic trees.  
In the spirit of \cite{Steel2006}, Matsen devises a neighbor joining algorithm, which, given dissimilarity maps $D_H, D_P$, 
finds a pair of trees $(T_H,T_P)$ within $k$ rooted SPR moves of one another; Matsen's algorithm attempts to minimize the sum of tree lengths $D_H \cdot b(T_H) + D_P \cdot b(T_P)$.  
\cite{erick} is the only work of its kind, and so far has not been extended to other spaces of cophylogenetic trees, such as $k$-NNI or $k$-interval.  We believe this is an important direction 
for future research.  Moreover, joint BME polytopes have not yet been studied for any space of cophylogenetic trees, and we believe this is an interesting geometric problem in light of \cite{NJME}.  

\subsection{Balanced minimum coevolution}

Constrained joint BME can be regarded as a constrained optimization problem:  Minimize the sum of tree lengths $\ell(T_H) + \ell(T_P)$, subject to $(T_H, T_P) \in S$.  Alternatively, given a 
distance measure $d(T_H, T_P)$ between trees or tree topologies, we can define the total {\em coevolution} $\ell(T_H, T_P)$ as $$\ell(T_H,T_P) := \ell(T_H) + \ell(T_P) + d(T_H,T_P),$$ where
$\ell(T_H)$ and $\ell(T_P)$ are defined as in standard BME.   We call the pair of trees $(T_H,T_P)$ which minimizes $\ell(T_H,T_P)$ the {\em balanced minimum coevolution} (BMC) cophylogeny.  
So far BMC has not been studied for any choice of distance measure between trees.  Natural choices of $d(T_H,T_P)$ for initial study might include SPR distance, NNI distance, and $k$-interval distance.
Also, analagous to \cite{Steel2006}, we can ask whether there is a fast heuristic like NJ for finding BMC cophylogenies.

\section{Cophylogenetic invariants}

Phylogenetic invariants are a well-studied subject in algebraic biology (\cite{allmanrhodes_chapter} and its references),  
and can be generalized to {\em cophylogenetic invariants}.

First, we would like to remind the reader about phylogenetic invariants.
Let $T$ be a rooted tree with $n$ leaves and let
$\mathcal{V}(T)$ be the set of nodes of $T$.  To each node $v
\in \mathcal{V}(T)$ let $X_v$ be a discrete random variable which takes
$k$ distinct states.
Consider the probability $P(X_v = i)$ that $X_v$ is in state $i$.
Let $\pi$ be a distribution of the
random variable $X_r$ at the root node $r$.  For each node $v \in
\mathcal{V}(T) \backslash \{r\}$, let $a(v)$ be the unique parent
of $v$. The transition from $a(v)$ to $v$ is given by a $k \times
k$-matrix $A^{(v)}$ of probabilities.  Then the probability
distribution at each node is computed recursively by the rule
\begin{equation}
\label{recursive P} P(X_v = j) \quad = \quad \sum_{i=1}^k
A^{(v)}_{ij} \cdot P(X_{a(v)}=i).
\end{equation}
This rule induces a joint distribution on all the random variables
$X_v$. We label the leaves of $T$ by
 $1,2,\ldots,n $, and we abbreviate
the marginal distribution on the variables at the leaves as
follows:
\begin{equation}
\label{unknownP}
 p_{i_1 i_2 \ldots i_n} \quad = \quad
P(X_{1} = i_1,X_{2} = i_2, \ldots, X_{n} = i_n ).
\end{equation}

A \emph{phylogenetic invariant} of the model is a polynomial in
the leaf probabilities $p_{i_1 i_2 \cdots i_n}$ which vanishes for
every choice of model parameters. The set of these polynomials
forms a prime ideal in the polynomial ring over the unknowns
$p_{i_1 i_2 \cdots i_n}$ (\cite{Allman2003a,Sturmfels2005} and 
their references).  Furthermore we have the following theorem.

\begin{theorem}[\cite{Sturmfels2005}]\label{thm2}
For any group based model on a phylogenetic tree $T$, the prime
ideal of phylogenetic invariants is generated by the invariants of
the local submodels around each interior node of $T$, together
with the quadratics which encode conditional independence statements
along the splits of $T$.
\end{theorem}

It is natural to ask whether invariants of cophylogenies can also be characterized.
Fix a group-based model for gene sequence evolution.  Let $p_{i_1 i_2 \cdots i_n}$ and $q_{i_1 i_2 \cdots i_n}$ be indeterminates 
representing leaf probabilities, for host and parasite sequences respectively.
For each tree topology $T_P$, let $I_{T_P} \subset R[p_{i_1 i_2 \cdots i_n}]$ be the ideal of phylogenetic invariants for $T_P$.
Similarly let $I_{T_H} \subset R[q_{i_1 i_2 \cdots i_n}]$ be the ideal of phylogenetic invariants for host topology $T_H$.

\begin{definition}
Given host and parasite tree topologies $T_H,T_P$, the {\em ideal of cophylogenetic invariants} for $T_P,T_H$ is the intersection ideal 
$R' I_{T_P} \cap R' I_{T_H} \subset R'$, where $R'= R[p_{i_1 i_2 \cdots i_n}, q_{i_1 i_2 \cdots i_n}]$.  
\end{definition}

\begin{definition}
Fix a space of cophylogenetic tree topologies $S \subset \mathcal{T}_H \times \mathcal{T}_P$.  Given a tree topology $T_H$, let $S_{T_H} = \{T_P \, | \, (T_H,T_P) \in S \}$.
The intersection $$J_{T_H} = \cap_{T_P \in S_{T_H}} I_{T_P}$$ is the {\em ideal of compatability invariants} for $T_H$ under $S$.
\end{definition}

Ideals of cophylogenetic invariants are easily computed from Theorem \ref{thm2}.  We hope that ideals of compatibility invariants can also be computed for particular 
spaces of cophylogenetic trees, without resorting to brute force computation of intersection ideals.

%\section{New geometries of tree space}
%{\bf K-INT, QUARTET, RF DISTANCE ARE EXPRESSIBLE BY INNER PRODUCT IN VS EMBEDDING.  CALLED A KERNEL ON ORIGINAL TREE (TOPOLOGY) SPACE.  
%NEW GEOMETRIES, SOME MIGHT TELL CLOUDS APART BY DIFF OF MEANS TESTING.  MORE GENERALLY, KERNEL-BASED METHODS LIKE SVM.}

\section{Open problems} \label{openprobs}
In this section we conclude by summarizing open problems.

%We showed that there are 5 possible parasite tree topologies for 
%the given host tree with 4 taxa under the $1$-interval cospeciation in Example
%\ref{ex3}.
\begin{problem}\label{pro1}
Is there a generalization of Theorem \ref{thm3} to $k > 1$?  In other words, given a host tree $T_H$, 
is there a combinatorial characterization of parasite trees under $k$-interval cospeciation?  How many parasite trees are possible for each host tree?
\end{problem}

\begin{problem}
Study the geometry and combinatorics of various spaces of cophylogenetic trees, in the spirit of \cite{Billera2001}.
\end{problem}

\begin{problem}\label{pro2}
Are there other interesting spaces of cophylogenetic trees which admit linear or tropical characterizations, like the extended Four Point Condition for the topology diagonal? 
\end{problem}

\begin{problem} 
Develop distance-based methods for cophylogenetic reconstruction, and study their robustness and geometric properties analagous to \cite{Mihaescu2006}.  
\end{problem}

\begin{problem}\label{pro6}
Can we compute and/or describe some generators of compatability ideals in a particular space of cophylogenetic trees, without resorting to a brute force computation of intersection ideals?
\end{problem}

\begin{problem}
Compute and study the face structures of joint BME polytopes for specific spaces of cophylogenetic trees.  In particular, can their graph structures (vertices and edges) be determined?  
\end{problem}

\section{Proof of Theorem~\ref{thm3}}
\label{sec:k_1_proof}

\begin{proof}
It is clear that if $T_P$ is at most 1 NNI from $T_H$, then $T_P$ and $T_H$ satisfy 1-interval cospeciation.  

To show that if there is 1-interval cospeciation, then $T_H$ and $T_P$ differ by at most one NNI, we use induction on the number of leaves.  The base case occurs when there are 4 leaves.  Then, there can be most one NNI, and thus the number of edges between leaves changes by at most 1.

If some cherry $(a,b)$ in $T_H$ that is also a cherry in $T_P$, then replace these cherries with the same leaf and we are done by induction.  It remains to consider the case when no cherry $(a,b)$ in $T_H$ is also a cherry in $T_P$.

1-interval cospeciation implies $a$ and $b$ are at most 3 edges apart, and thus exactly 3 edges apart, since they do not form a cherry.  Then without loss of generality, $a$ forms a cherry with some subtree $S_c$, and $b$ is attached just above this cherry, as shown in Fig.~\ref{fig:k_1_proof}.  Futhermore, we can assume, without loss of generality, that in a sequence of NNIs transforming $T_H$ into $T_P$, the NNI moving $S_c$ to between $a$ and $b$ occurs last.  Let $T_c$ be the second last tree in this sequence of NNIs.  Then by our hypothesis, $T_H$ and $T_c$ differ by at least 1 NNI.
\begin{figure}[ht]
\centering
\includegraphics[scale = 0.5]{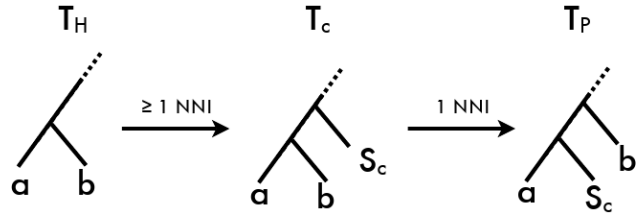}
\caption{The trees used in the proof of Theorem~\ref{thm3}.}%\ref{th:k1}.}
\label{fig:k_1_proof}
\end{figure}
If $T_H$ and $T_P$ differ by exactly 1 NNI, then there are two cases.  If this NNI is done about an edge in $S_c$, then some leaf $l$ in $S_c$ moves one edge closer to the root of $S_c$.  This implies the number of edges between $a$ and $l$ decreases by 2 between $T_H$ and $T_P$, which is a contradiction.  Otherwise, $a$, $b$, and $S_c$ are all contained in one of the subtrees involved in the interchange, say $A$ (using the notation from Fig.~\ref{fig:NNIoperation}).  Then this NNI moves $a$ one edge closer to each leaf in $B$, and the second NNI moves $a$ one edge closer to the root of $A$, and hence to each leaf in $B$.  Thus, the number of edges between $a$ and any edge in $B$ decreases by 2 between $T_H$ and $T_P$, which is a contradiction.

Thus at least two NNIs are needed to transform $T_H$ into $T_c$.  Both $T_H$ and $T_c$ contain the cherry $(a,b)$, so replace this cherry with the leaf $ab$ in $T_H$ to get the tree $T_H'$, and in $T_c$ to get the tree $T_c'$.  

By the induction hypothesis on $T_H'$ and $T_c'$, there are distinct leaves $i,j$ such that 
\begin{equation}
\label{eq:leaf_conditions}
\abs{ d_{T_c'} (i,j) - d_{T_H'}(i,j) } > 1.
\end{equation}
\\
\emph{Case ($i, j \notin S_c$ and $i, j \neq ab$) or ($i, j \in S_c$) : }
Then $d_{T_c'}(i,j) = d_{T_c}(i,j) = d_{T_P}(i,j)$ and $d_{T_H'}(i,j) = d_{T_H}(i,j)$.  Plugging into \eqref{eq:leaf_conditions} gives $\abs{ d_{T_P} (i,j) - d_{T_H}(i,j) } > 1$. \\
\\
\emph{Case $i \notin S_c$, and $j = ab$ : }
Then $d_{T_c'}(i,ab) = d_{T_c}(i, a) - 1 = d_{T_P}(i,a) -1$, and $d_{T_H'}(i, ab) = d_{T_H}(i,a) - 1$.  Plugging into \eqref{eq:leaf_conditions} gives $\abs{ d_{T_P} (i,a) - d_{T_H}(i,a) } > 1$. \\
\\
\emph{Case $i \in S_c$, $j \notin S_c$, and $j \neq ab$ :}
Consider the subtree of $T_c'$ containing $ab$, $S_c$, $j$, and the paths between them.  
Let $x$ be the number of edges between $j$ and the common ancestor of $ab$ and $i$.  
Let $y$ be the number of edges from the root of subtree $S_c$ to $i$.  
Let $V$ be the interior vertex where the paths from  $ab$, $i$, and $j$ meet in $T_H'$.  
Let $u$ be the number of edges between $ab$ and $V$.  
Let $v$ be the number of edges between $j$ and $V$.  
Let $w$ be the number of edges between $i$ and $V$.
Then $d_{T_c}(i,j) = 1 + x + y$.  Since $d_{T_c'}(i,j) = d_{T_c}(i,j)$ and $d_{T_H'}(i,j) = d_{T_H}(i,j)$, then $\abs{ d_{T_H}(i,j) -  d_{T_c}(i,j) }>1$, which implies $d_{T_H}(i,j) \leq x + y - 1$ or $d_{T_H}(i,j) \geq x + y + 3$.  Now $d_{T_P}(i,j) = 2 + x + y$, so 1-interval cospeciation implies $d_{T_H}(i,j) = x + y + 3$.  By definition of $v$ and $w$, $v + w = d_{T_H}(i,j) = x + y + 3$.   

We have $d_{T_P}(a,i) = 2 + y$, $d_{T_P}(b,i) = 3 + y$, and $d_{T_H}(a,i) = d_{T_H}(b,i)$.  So 1-interval cospeciation implies $d_{T_H}(a,i) = d_{T_H}(b,i) = 2 + y \text{ or } 3+ y$.  This implies $u + w = 1 + y \text{ or } 2 + y$.  We also have $d_{T_P}(a,j) = 2 + x$, $d_{T_P}(b,j) = 1 + x$, and $d_{T_H}(a,j) = d_{T_H}(b,j)$.  Then 1-interval cospeciation implies $d_{T_H}(a,j) = d_{T_H}(b,j) = 1+ x \text{ or } 2+x$.  This implies $u + v = x \text{ or } 1 + x$.

Then either $(v+w) -(u + w) = v - u = x+ 2$ or $x + 1$.  If $v - u = x+ 2$ and $u + v = x$, then $2v = 2x + 2$ or $v = x+1$.  This implies $u = -1$, which is impossible.  If $v - u = x+ 2$ and $u + v = 1 + x$, then $2v = 2x + 3$, which is impossible, because all variables are numbers of edges, and hence integers.  If $v - u = x + 1$ and $u + v = x$, then $2v = 2x + 1$, which is also impossible because all variables are integers.  Finally, if $v - u = x + 1$ and $u + v = x$, then $2v = 2x + 2$, which implies $v = x + 1$ and $u = -1$, which is impossible.  Therefore, $v + w = d_{T_H}(i,j) \neq x + y + 3$, which is a contradiction. \\
\\
\emph{Case $i = ab$ and $j \in S_c$ :}
Then $d_{T_c'}(ab,j) = d_{T_c}(b,j) - 1 = d_{T_P}(b,j) -1$ and $d_{T_H'}(ab,j) = d_{T_H}(b,j) - 1$.  Plugging into \eqref{eq:leaf_conditions} gives $\abs{ d_{T_P} (b,j) - d_{T_H}(b,j) } > 1$. \\
\\

Therefore, there exist at least two leaves, such that the number of edges between them changes by more than 1 from $T_H$ to $T_P$, which is a contradiction.
\end{proof}

\section*{Acknowledgements}
P. Huggins is supported by the Lane Fellowship in computational biology.  R. Yoshida is supported by NIH R01 grant 1R01GM086888-01.  
This material was based upon work partially supported by the National
Science Foundation under Grant DMS-0635449 to the Statistical and
Applied Mathematical Sciences Institute. Any opinions, findings, and
conclusions or recommendations expressed in this material are those of
the author(s) and do not necessarily reflect the views of the National
Science Foundation. 
% added the acknowledgement for me being supported by SAMSI.  (Megan)

%
% ---- Bibliography ----
%

%
\end{document}